\documentclass[10pt,conference,letterpaper]{IEEEtran}
\IEEEoverridecommandlockouts
\usepackage{cite}
\usepackage{amsmath,amssymb,amsfonts}
\usepackage{graphicx}
\usepackage{textcomp}
\usepackage{xcolor}
\def\BibTeX{{\rm B\kern-.05em{\sc i\kern-.025em b}\kern-.08em
    T\kern-.1667em\lower.7ex\hbox{E}\kern-.125emX}}

\usepackage{booktabs}
\usepackage{algorithm}
\usepackage[noend]{algorithmic}
\usepackage{scalefnt}
\usepackage[hidelinks]{hyperref}
\usepackage[utf8]{inputenc}
\usepackage{tikz}
\usetikzlibrary{positioning, arrows, arrows.meta, trees, shapes.geometric}
\usepackage{adjustbox}
\usepackage{subcaption}
\usepackage{multirow}
\usepackage{amsthm}

\newtheorem{remark}{Remark}

\newtheorem{proposition}{Proposition}
\newtheorem{property}{Property}
\def\prob {$\mathsf{PROB}$}
\def\inc {\ensuremath{\mathsf{INC}}}
\def\incap {\ensuremath{\mathsf{INC_{AP}}}}

\def\smooth {\ensuremath{\mathsf{Smooth}}}

\newcommand{\lo}[1]{\mathsf{Lo}(#1)}
\newcommand{\hi}[1]{\mathsf{Hi}(#1)}

\newcommand{\child}[1]{\mathsf{Child}(#1)}
\newcommand{\parent}[1]{\mathsf{Parent}(#1)}
\newcommand{\varset}[1]{\mathsf{VarSet}(#1)}
\newcommand{\var}[1]{\mathsf{Var}(#1)}
\newcommand{\subdiagram}[1]{\mathsf{Subdiagram}(#1)}

\begin{document}

\title{INC: A Scalable Incremental Weighted Sampler}

\author{
\IEEEauthorblockN{Suwei Yang}
\IEEEauthorblockA{
\textit{National University of Singapore}\\
Singapore
}
\and
\IEEEauthorblockN{Victor C. Liang}
\IEEEauthorblockA{\textit{GrabTaxi Holdings Pte. Ltd.} \\
Singapore
}
\and
\IEEEauthorblockN{Kuldeep S. Meel}
\IEEEauthorblockA{
\textit{National University of Singapore}\\
Singapore
}
}

\maketitle

\begin{abstract}
    The fundamental problem of weighted sampling involves sampling of satisfying assignments of Boolean formulas, which specify sampling sets, and according to distributions defined by pre-specified weight functions to weight functions. The tight integration of sampling routines in various applications has highlighted the need for samplers to be incremental, i.e., samplers are expected to handle updates to weight functions. 

    The primary contribution of this work is an efficient knowledge compilation-based weighted sampler, {\inc}\footnote{code available at \url{https://github.com/grab/inc-weighted-sampler/}}, designed for incremental sampling. {\inc} builds on top of the recently proposed knowledge compilation language, OBDD[$\wedge$], and is accompanied by rigorous theoretical guarantees. Our extensive experiments demonstrate that {\inc} is faster than state-of-the-art approach for majority of the evaluation. In particular, we observed a median of $1.69\times$ runtime improvement over the prior state-of-the-art approach.
\end{abstract}


\section{Introduction} \label{sec:introduction}

Given a Boolean formula $F$ and weight function $W$, weighted sampling involves sampling from the set of satisfying assignments of $F$ according to the distribution defined by $W$. Weighted sampling is a fundamental problem in many fields such as computer science, mathematics and physics, with numerous applications. In particular, constrained-random simulation forms the bedrock of modern hardware and software verification efforts~\cite{KK07}.

Sampling techniques are fundamental building blocks, and there has been sustained interest in the development of sampling tools and techniques. Recent years witnessed the introduction of numerous sampling tools and techniques, from approximate sampling techniques to uniform samplers SPUR and KUS, and weighted sampler WAPS~\cite{JS1996,SSL15,AHT18,SGRM18,GSRM19}. Sampling tools and techniques have seen continuous adoption in many applications and settings~\cite{NRJKVMS07,GPMXWOCB14,KW19,BLM20,PLVSMMVKG20,BCMS21}. The scalability of a sampler is a consideration that directly affects its adoption rate. Therefore, improving scalability continues to be a key objective for the community focused on developing samplers.

The tight integration of sampling routines in various applications has highlighted the importance for samplers to handle incremental weight updates over multiple sampling rounds, also known as incremental weighted sampling. Existing efforts on improving scalability typically focus on single round weighted sampling, and might have overlooked the incremental setting. In particular, existing approaches involving incremental weighted sampling typically employ off-the-shelf weighted samplers which could lead to less than ideal incremental sampling performance.

The primary contribution of this work is an efficient scalable weighted sampler {\inc} that is designed from the ground up to address scalability issues in incremental weighted sampling settings. The core architecture of {\inc} is based on knowledge compilation (KC) paradigm, which seeks to succinctly represent all satisfying assignments of a Boolean formula with a directed acyclic graph (DAG)~\cite{DM02}. In the design of {\inc}, we make two core decisions that are responsible for outperforming the current state-of-the-art weighted sampler. Firstly, we build {\inc} on top of {\prob} (Probabilistic OBDD[$\land$]~\cite{LLY17}) which is substantially smaller than the KC diagram used in the prior state-of-the-art approaches. Secondly, {\inc} is designed to perform \textit{annotation}, which refers to the computation of joint probabilities, in log-space to avoid the slower alternative of using arbitrary precision math computations.

Given a Boolean formula $F$ and weight function $W$, {\inc} compiles and stores the compiled {\prob} in the first round of sampling. The weight updates for subsequent incremental sampling rounds are processed without recompilation, amortizing the compilation cost. 
Furthermore, for each sampling round, {\inc} simultaneously performs \textit{annotation} and sampling in a single bottom-up pass of the {\prob}, achieving speedup over existing approaches. We observed that {\inc} is significantly faster than the existing state-of-the-art in the incremental sampling routine. In our empirical evaluations, {\inc} achieved a median of $1.69\times$ runtime improvement over the state-of-the-art weighted sampler, WAPS~\cite{GSRM19}. Additional performance breakdown analysis supports our design choices in the development of {\inc}. In particular, {\prob} is on median $4.64\times$ smaller than the KC diagram used by the competing approach, and log-space \textit{annotation} computations are on median $1.12\times$ faster than arbitrary precision computations. Furthermore, {\inc} demonstrated significantly better handling of incremental sampling rounds, with incremental sampling rounds to be on median $5.9\%$ of the initial round, compared to $67.6\%$ for WAPS. 

The rest of the paper is organized as follows. We first introduce the relevant background knowledge and related works in Section~\ref{sec:background}. We then introduce {\prob} and its properties in Section~\ref{sec:prob}. In Section~\ref{sec:inc}, we introduce our weighted sampler {\inc}, detail important implementation decisions, and provide theoretical analysis of {\inc}. We then describe the extensive empirical evaluations and discuss the results in Section~\ref{sec:experiments}. Finally, we conclude in Section~\ref{sec:future-works}.

\section{Background and Related Work} \label{sec:background}

\paragraph*{\textbf{Knowledge Compilation}}

Knowledge compilation (KC) involves representing logical formulas as directed acyclic graphs (DAG), which are commonly referred to as knowledge compilation diagrams~\cite{DM02}. The goal of knowledge compilation is to allow for tractable computation of certain queries such as model counting and weighted sampling. There are many well-studied forms of knowledge compilation diagrams such as d-DNNF, SDD, BDD, ZDD, OBDD, AOBDD, and the likes~\cite{L1959,B1986,M1993,A01,A02,MDM08,A11}. In this work, we build our weighted sampler upon a variant of OBDD known as OBDD[$\land$]~\cite{LLY17}.

\paragraph*{\textbf{OBDD[$\land$]}}

Lee~\cite{L1959} introduced Binary Decision Diagram (BDD) as a way to represent Shannon expansion~\cite{B1854}. \cite{B1986} introduced fixed variable orderings to BDDs (known as OBDD)~\cite{B1986} for canonical representation and compression of BDDs via shared sub-graphs. Lai et al.~\cite{LLY17} introduced conjunction nodes to OBDDs (known as OBDD[$\land$])~\cite{LLY17} to further reduce the size of the resultant DAG to represent a given Boolean formula. In this work, we parameterize an OBDD[$\land$] to form a {\prob} that is used for weighted sampling.

\paragraph*{\textbf{Sampling}}

A Boolean variable $x$ can be assigned either \textit{true} or \textit{false}, and its literal refers to either $x$ or its negation. A Boolean formula is in conjunctive normal form (CNF) if it is a conjunction of clauses, with each clause being a disjunction of literals. A Boolean formula $F$ is satisfiable if there exists an assignment $\tau$ of its variables such that the $F$ evaluates to \textit{true}. The model count of Boolean formula $F$ refers to the number of distinct satisfying assignments of $F$.

Weighed sampling concerns with sampling elements from a distribution according to non-negative weights provided by a user-defined weight function $W$. In the context of this work, weighted sampling refers to the process of sampling from the space of satisfying assignments of a Boolean formula $F$. The weight function $W$ assigns a non-negative weight to each literal $l$ of $F$. The weight of an assignment $\tau$ is defined as the product of the weight of its literals. 

\paragraph*{\textbf{WAPS}}

KUS~\cite{SGRM18} utilizes knowledge compilation techniques, specifically Deterministic Decomposable Negation Normal Form (d-DNNF)~\cite{A02}, to perform uniform sampling in 2 passes of the d-DNNF. \textit{Annotation} is performed in the first pass, followed by sampling. WAPS~\cite{GSRM19} improves upon KUS by enabling weighted sampling via parameterization of the d-DNNF. WAPS performs sampling in a similar manner to KUS, the main difference being that the \textit{annotation} step in WAPS takes into account the provided weight function. In contrast, we introduce {\inc} which performs weighted sampling in a single pass by leveraging the DAG structure of {\prob}.

Knowledge compilation-based samplers typically perform incremental sampling as follows. The sampling space is first expressed as satisfying assignments of a Boolean formula, which is then compiled into the respective knowledge compilation form. In the following step, samples are drawn according to the given weight function $W$. Subsequently, the weights are updated depending on application logic and weighted sampling is performed again. The process is repeated until an application-specific stopping criterion is met. An example of such an application would be the Baital framework~\cite{BLM20}, developed to use incremental weighted sampling to generate test cases for configurable systems.

\section{{\prob}: - Probabilistic OBDD[$\land$]} \label{sec:prob}

{\prob} is a DAG composed of four types of nodes - \textit{conjunction}, \textit{decision}, \textit{true} and \textit{false} nodes. The internal nodes of a {\prob} consist of conjunction and decision nodes whereas the leaf nodes of the {\prob} consist of true and false nodes. A {\prob} is recursively made up of sub-{\prob}s that represent sub-formulas of Boolean formula $F$. We use $\varset{n}$ to refer to the set of variables of $F$ represented by a {\prob} with $n$ as the root node. $\subdiagram{n}$ refers to the sub-{\prob} starting at node $n$ and $\parent{n}$ refers to the immediate parent of node $n$ in {\prob}.

\subsection{{\prob} Structure}

\paragraph*{\textbf{Conjunction node ($\land$-node)}}
A $\land$-node $n_c$ represents conjunctions in the assignment space. There are no limits to the number of child nodes that $n_c$ can have. However, the set of variables ($\varset{\cdot}$) of each child node of $n_c$ must be disjoint. An example of a $\land$-node would be $n2$ in Figure~\ref{fig:smooth-prob-example}. Notice that $\varset{n4} = \{z\}$ and $\varset{n5} = \{y\}$ are disjoint.

\paragraph*{\textbf{Decision node}}
A decision node $n_d$ represents decisions on the associated Boolean variable $\var{n_d}$ in Boolean formula $F$ that the {\prob} represents. A decision node can have exactly two children - \textit{lo-child} ($\lo{n_d}$) and \textit{hi-child} ($\hi{n_d}$). $\lo{n_d}$ represents the assignment space when $\var{n_d}$ is set to \textit{false} and $\hi{n_d}$ represents otherwise. $\theta_{n_{d_{hi}}}$ and $\theta_{n_{d_{lo}}}$ refer to the parameters associated with the edge connecting decision node $n_d$ with $\hi{n_d}$ and $\lo{n_d}$ respectively in a {\prob}. Node $n1$ in Figure~\ref{fig:smooth-prob-example} is a decision node with $\var{n1} = x$, $\hi{n1} = n3$ and $\lo{n1} = n2$.

\paragraph*{\textbf{True and False nodes}}
True ($\top$) and false ($\bot$) nodes are leaf nodes in a {\prob}. Let $\tau$ be an assignment of all variables of Boolean formula $F$ and let {\prob} $\psi$ represent $F$. $\tau$ corresponds to a traversal of $\psi$ from the root node to leaf nodes. The traversal follows $\tau$ at every decision node and visits all child nodes of every conjunction node encountered along the way. $\tau$ is a satisfying assignment if all parts of the traversal eventually lead to the true node. $\tau$ is not a satisfying assignment if any part of the traversal leads to the false node. With reference to Figure~\ref{fig:smooth-prob-example}, let $\tau_1 = \{x, y, \neg z\}$ and $\tau_2 = \{x, y, z\}$. For $\tau_1$, the traversal would visit $n1, n3, n6, n7, n9$, and $\tau_1$ is a satisfying assignment since the traversal always leads to $\top$ node ($n9$). As a counter-example, $\tau_2$ is not a satisfying assignment with its corresponding traversal visiting $n1, n3, n6, n7, n8, n9$. $\tau_2$ traversal visits $\bot$ node ($n8$) because variable $z \mapsto \textit{true}$ in $\tau_2$ and $\hi{n6}$ is node $n8$.

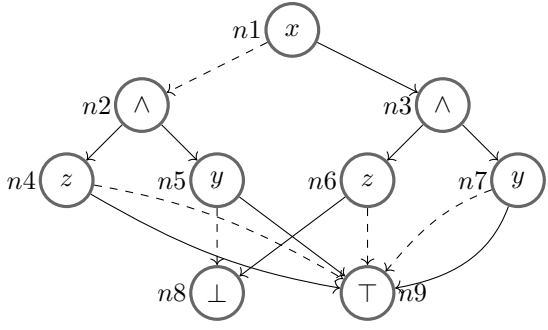
\begin{figure}[htb]
    \centering
    \begin{tikzpicture}[
    roundnode/.style={circle, draw=black!60, very thick, minimum size=7mm},
    ]
        \node[roundnode](x) at (0, 0){$x$};
        \node[roundnode](left-x-and) at (-2, -1){$\land$};
        \node[roundnode](right-x-and) at (2, -1){$\land$};
        \node[roundnode](y) at (-1, -2){$y$};
        \node[roundnode](z) at (1, -2){$z$};
        \node[roundnode](z-c) at (-3, -2){$z$};
        \node[roundnode](y-c) at (3, -2){$y$};
        \node[roundnode](t) at (1, -3.5){$\top$};
        \node[roundnode](f) at (-1, -3.5){$\bot$};

        \draw [dashed,->] (x) -- (left-x-and);
        \draw [->] (x) -- (right-x-and);
        \draw [->] (y) -- (t);
        \draw [dashed,->] (y) -- (f);
        \draw [dashed,->] (z) -- (t);
        \draw [->] (z) -- (f);
        \draw [->] (left-x-and) -- (y);
        \draw [->] (left-x-and) -- (z-c);
        \draw [->] (right-x-and) -- (z);
        \draw [->] (right-x-and) -- (y-c);
        \draw [->] (z-c) to[out=-30,in=170] (t);
        \draw [dashed,->] (z-c) to[out=-10,in=150] (t);
        \draw [->] (y-c) to[out=-110,in=10] (t);
        \draw [dashed,->] (y-c) to[out=-160,in=50] (t);

        \node at (-0.6, 0) {$n1$};
        \node at (-2.6, -1){$n2$};
        \node at (1.4, -1){$n3$};
        \node at (-3.6, -2){$n4$};
        \node at (-1.6, -2){$n5$};
        \node at (0.4, -2){$n6$};
        \node at (2.4, -2){$n7$};
        \node at (-1.6, -3.5){$n8$};
        \node at (1.6, -3.5){$n9$};
    \end{tikzpicture}
    \caption{A smooth {\prob} $\psi_1$ with 9 nodes, $n1, ... , n9$, representing $F = (x \lor y) \land (\neg x \lor \neg z)$. Branch parameters are omitted}
    \label{fig:smooth-prob-example}
\end{figure}

\subsection{{\prob} Parameters} \label{subsec:prob-param}

In the {\prob} structure, each decision node $n_d$ has two parameters $\theta_{\lo{n_d}}$ and $\theta_{\hi{n_d}}$, associated with the two branches of $n_d$, which sums up to $1$. $\theta_{\lo{n_d}}$ is the normalized weight of the literal $\neg \var{n_d}$ and similarly, $\theta_{\hi{n_d}}$ is that of the literal $\var{n_d}$. One can view $\theta_{\lo{n_d}}$ to be the probability of picking $\neg \var{n_d}$ and $\theta_{\hi{n_d}}$ to be that of picking $\var{n_d}$ by the \textit{determinism} property introduced later. Let $x_i$ be $\var{n_d}$. Given a weight function $W$:

\begin{align*}
    & \theta_{\lo{n_d}} = \frac{W(\neg x_i )}{W(\neg x_i ) + W(x_i)} &  \theta_{\hi{n_d}} = \frac{W(x_i)}{W(\neg x_i ) + W(x_i)}
\end{align*}

\subsection{{\prob} Properties}

The {\prob} structure has important properties such as \textit{determinism} and \textit{decomposability}. In addition to the \textit{determinism} and \textit{decomposability} properties, we ensure that {\prob}s used in this work have the \textit{smoothness} property through a smoothing process (Algorithm~\ref{alg:smooth}).

\begin{property}[Determinism]
    For every decision node $n_d$, the set of satisfying assignments represented by $\hi{n_d}$ and $\lo{n_d}$ are logically disjoint.
\end{property}

\begin{property}[Decomposability]
    For every conjunction node $n_c$, $\varset{c_i} \cap \varset{c_j} = \emptyset$ for all $c_i$ and $c_j$ where $c_i, c_j \in \child{n_c}$ and $c_i \not= c_j$. 
\end{property}
    
\begin{property}[Smoothness]
    For every decision node $n_d$, $\varset{\hi{n_d}} = \varset{\lo{n_d}}$.
\end{property}

\subsection{Joint Probability Calculation with {\prob}}

In Section~\ref{subsec:prob-param}, we mention that one can view the branch parameters as the probability of choosing between the positive and negative literal of a decision node. Notice that because of the \textit{decomposability} and \textit{determinism} properties of {\prob}, it is straightforward to calculate the joint probabilities at every given node. At each conjunction node $n_c$, since the variable sets of the child nodes of $n_c$ are disjoint by \textit{decomposability}, the joint probability of $n_c$ is simply the product of joint probabilities of each child node. At each decision node $n_d$, there are only two possible outcomes on $\var{n_d}$ - positive literal $\var{n_d}$ or negative literal $\neg \var{n_d}$. By \textit{determinism} property, the joint probability is the sum of the two possible scenarios. Formally, the calculations for joint probabilities $P'$ at each node in {\prob} are as follows:

\begin{align}
    P'\text{ of $\land$-node }n_c &= \prod_{c \in \child{n_c}} P'(c) \label{eq:eq1}\tag{EQ1}
    \\
    P'\text{ of decision-node }n_d &= \theta_{\lo{n_d}} \times P'(\lo{n_d}) \label{eq:eq2}\tag{EQ2}
    \\
    & + \theta_{\hi{n_d}} \times P'(\hi{n_d}) \notag
\end{align}

For true node $n$, $P'(n)=1$ because it represents satisfying assignments when reached. In contrast $P'(n)=0$ when $n$ is a \textit{false} node as it represents non-satisfying assignments. In Proposition~\ref{prop:sampling-correctness}, we show that weighted sampling is equivalent to sampling according to joint probabilities of satisfying assignments of a {\prob}.

\section{{\inc} - Sampling from {\prob}} \label{sec:inc}

In this section, we introduce {\inc} - a bottom-up algorithm for weighted sampling on {\prob}. We first describe {\inc} for drawing one sample and subsequently describe how to extend {\inc} to draw $k$ samples at once. We also provide proof of correctness that {\inc} is indeed performing weighted sampling. As a side note, samples are drawn with replacement, in line with the existing state-of-the-art weighted sampler~\cite{GSRM19}.

\subsection{Preprocessing {\prob}}

In the main sampling algorithm (Algorithm~\ref{alg:inc-sampling}) to be introduced later in this section, the input is a smooth {\prob}. As a preprocessing step, we introduce {\smooth} algorithm that takes in a {\prob} $\psi$ and performs smoothing.

\begin{algorithm}[htb]
    \begin{flushleft}
        \textbf{Input}: {\prob} $\psi$\\
        \textbf{Output}: smooth {\prob}
    \end{flushleft}
    \begin{algorithmic}[1] %
        \STATE $\kappa \leftarrow \mathsf{initMap()}$
        \FOR{node $n$ of $\psi$ in bottom-up order}
            \IF{$n$ is \textit{true} node or \textit{false} node} 
                \STATE $\kappa[n] \leftarrow \emptyset$
            \ELSIF{$n$ is $\land$-node}
                \STATE $\kappa[n] \leftarrow$ $\mathsf{unionVarSet}$(${\child{n}}, \kappa$)
            \ELSE
                \IF{$\kappa[{\hi{n}}] - \kappa[{\lo{n}}] \neq \emptyset$} \label{line:smooth-lo}
                    \STATE lset $\leftarrow \kappa[{\hi{n}}] - \kappa[{\lo{n}}]$
                    \STATE lcNode $\leftarrow new\land-node()$
                    \STATE lcNode.$\mathsf{addChild({\lo{n}})}$
                    \FOR{var $v$ in lset}
                        \STATE dNode $\leftarrow \mathsf{checkMakeTrueDecisionNode}(v)$
                        \STATE lcNode.$\mathsf{addChild}$(dNode)
                    \ENDFOR
                    \STATE ${\lo{n}} 
                    \leftarrow $lcNode
                \ENDIF
                \IF{$\kappa[{\lo{n}}] - \kappa[{\hi{n}}] \neq \emptyset$} \label{line:smooth-hi}
                    \STATE rset $\leftarrow \kappa[{\lo{n}}] - \kappa[{\hi{n}}]$
                    \STATE rcNode $\leftarrow new\land-node()$
                    \STATE rcNode.$\mathsf{addChild({\hi{n}})}$
                    \FOR{var $v$ in rset}
                        \STATE dNode $\leftarrow \mathsf{checkMakeTrueDecisionNode}(v)$
                        \STATE rcNode.$\mathsf{addChild}$(dNode)
                    \ENDFOR
                    \STATE ${\hi{n}} 
                    \leftarrow $rcNode
                \ENDIF
                \STATE $\kappa[n] \leftarrow {\var{n}} \cup \mathsf{unionVarSet}$($\{{\hi{n}}, {\lo{n}}\}$)
            \ENDIF
        \ENDFOR
        \STATE \textbf{return} $\psi$
    \end{algorithmic}
    \caption{$\mathsf{Smooth}$ - returns a smoothed {\prob}}
    \label{alg:smooth}
\end{algorithm}

The $\mathsf{Smooth}$ algorithm processes the nodes in the input {\prob} $\psi$ in a bottom-up manner while keeping track of $\varset{n}$ for every node n in $\psi$ using a map $\kappa$. \textit{True} and \textit{false} nodes have $\emptyset$ as they are leaf nodes and do not represent any variables. At each conjunction node, its variable set is the union of variable sets of its child nodes.

The smoothing happens at decision node $n$ in $\psi$ when $\varset{\lo{n}}$ and $\varset{\hi{n}}$ do not contain the same set of variables as shown by lines~\ref{line:smooth-lo} and~\ref{line:smooth-hi} of Algorithm~\ref{alg:smooth}. In the smoothing process, a new conjunction node (\textit{lcNode} for ${\lo{n}}$ and \textit{rcNode} for ${\hi{n}}$) is created to replace the corresponding child of $n$, with the original child node now set as a child of the conjunction node. Additionally, for each of the missing variables $v$, a decision node representing $v$ is created and added as a child of the respective conjunction node. The decision nodes created during smoothing have both their lo-child and hi-child set to the \textit{true} node. To reduce memory footprint, we check if there exists the same decision node before creating it in the $\mathsf{checkMakeTrueDecisionNode}$ function.

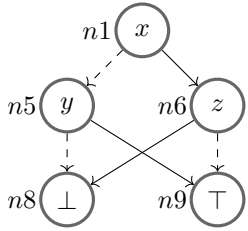
\begin{figure}[htb]
    \centering
    \begin{tikzpicture}[
    roundnode/.style={circle, draw=black!60, very thick, minimum size=7mm},
    ]
    \node[roundnode](x) at (0, 0){$x$};
    \node[roundnode](y) at (-1, -1){$y$};
    \node[roundnode](z) at (1, -1){$z$};
    \node[roundnode](t) at (1, -2.25){$\top$};
    \node[roundnode](f) at (-1, -2.25){$\bot$};

    \draw [dashed,->] (x) -- (y);
    \draw [->] (x) -- (z);
    \draw [->] (y) -- (t);
    \draw [dashed,->] (y) -- (f);
    \draw [dashed,->] (z) -- (t);
    \draw [->] (z) -- (f);

    \node at (-0.6, 0) {$n1$};
    \node at (-1.6, -1) {$n5$};
    \node at (0.4, -1) {$n6$};
    \node at (-1.6, -2.25) {$n8$};
    \node at (0.4, -2.25) {$n9$};

    \end{tikzpicture}
    \caption{A {\prob} $\psi_2$ representing Boolean formula $F = (x \lor y) \land (\neg x \lor \neg z)$, branch parameters are omitted}
    \label{fig:Prob-non-smooth}
\end{figure}

As an example, we refer to $\psi_2$ in Figure~\ref{fig:Prob-non-smooth}. It is obvious that $\psi_2$ is not smooth, because $\varset{\lo{n1}} = \{y\}$ and $\varset{\hi{n1}} = \{z\}$. In the smoothing process, we replace ${\lo{n1}}$ with a new conjunction node $n2$ and add a decision node $n4$ representing missing variable $z$, with both child set to \textit{true} node $n9$. We repeat the steps for ${\hi{n1}}$ to arrive at {\prob} $\psi_1$ in Figure~\ref{fig:smooth-prob-example}.

\subsection{Sampling Algorithm}

\begin{algorithm}[htb]
    \begin{flushleft}
        \textbf{Input}: smooth {\prob} $\psi$\\
        \textbf{Output}: a sampled satisfying assignment
    \end{flushleft}
    \begin{algorithmic}[1] %
        \STATE cache $\omega$ $\leftarrow$ $\mathsf{initCache()}$
        \STATE joint prob cache $\varphi$ $\leftarrow$ $\mathsf{initCache()}$
        \STATE $\psi'$ $\leftarrow$ $\mathsf{hideFalseNode}(\psi)$ \label{line:zddc-preprocess-roulette} 
        \FOR{node $n$ of $\psi'$ in bottom-up order}
            \IF{$n$ is \textit{true} node} 
                \STATE $\omega[n] \leftarrow \emptyset$
                \STATE $\varphi[n] \leftarrow 1$ \label{line:true-node-calc}
            \ELSIF{$n$ is $\land$-node}
                \STATE $\omega[n] \leftarrow$ $\mathsf{unionChild}$(${\child{n}}, \omega$) \label{line:union-conjunction-child-roulette}
                \STATE $\varphi[n] \leftarrow \prod_{c \in {\child{n}}} \varphi[c]$ \label{line:conjunction-node-calc}
            \ELSE
                \STATE $p_{lo} \leftarrow \theta_{\lo{n}} \times \varphi[{\lo{n}}]$ \label{line:sample-decision-start-roulette}
                \STATE $p_{hi} \leftarrow \theta_{\hi{n}} \times \varphi[{\hi{n}}]$
                \STATE $p_{joint} \leftarrow p_{lo} + p_{hi}$
                \STATE $\varphi[n] \leftarrow p_{joint}$
                \STATE $r$ $\leftarrow$ $x \sim \mathsf{binomial}(1, \frac{p_{hi}}{p_{joint}})$ \label{line:binomial-random}
                \IF{$r$ is $1$} 
                    \STATE $\omega[n]$ $\leftarrow$ $\omega[\hi{n}] \cup \var{n}$
                \ELSE
                    \STATE $\omega[n]$ $\leftarrow$ $\omega[\lo{n}] \cup \neg \var{n}$
                \ENDIF \label{line:sample-decision-end-roulette}
            \ENDIF
        \ENDFOR
        \STATE \textbf{return} $\omega$[rootnode($\psi$)]
    \end{algorithmic}
    \caption{{\inc} - returns a satisfying assignment based on {\prob} $\psi$ parameters}
    \label{alg:inc-sampling}
\end{algorithm}

{\inc} takes a {\prob} $\psi$ representing Boolean formula $F$ and draws a sample from the space of satisfying assignments of $F$, the process is illustrated by Algorithm~\ref{alg:inc-sampling}. {\inc} performs sampling in a bottom-up manner while integrating the \textit{annotation} process in the same bottom-up pass. Since we want to sample from the space of satisfying assignments we can ignore \textit{false} nodes in $\psi$ entirely by considering a sub-DAG that excludes \textit{false} nodes and edges leading to them, as shown by line~\ref{line:zddc-preprocess-roulette}. As an example, $\mathsf{hideFalseNode}$ when applied to $\psi_1$ would remove node $n8$ and the edges immediately leading to it. Next, {\inc} processes each of the remaining nodes in bottom-up order while keeping two caches - $\omega$ to store the partial samples from each node, $\varphi$ to store the joint probability at each node. {\inc} starts with $\emptyset$ at the \textit{true} node since there is no associated variable.

At each conjunction node, {\inc} takes the union of the child nodes in line~\ref{line:union-conjunction-child-roulette}. Using $n2$ in Figure~\ref{fig:smooth-prob-example} as an example, if sample drawn at $n4$ is $\omega[n4] = \{\neg z\}$ and at $n5$ is $\omega[n5] = \{y\}$, then $\mathsf{unionChild}(\child{n2},\omega) = \{y, \neg z\}$. At each decision node $n$, a decision on $\var{n}$ is sampled from lines~\ref{line:binomial-random} to~\ref{line:sample-decision-end-roulette}. We first calculate the joint probabilities, $p_{lo}$ and $p_{hi}$ of choosing $\neg \var{n}$ and choosing $\var{n}$. Subsequently, we sample decision on $\var{n}$ using a binomial distribution in line~\ref{line:binomial-random} with the probability of success being the joint probability of choosing $\var{n}$. After processing all nodes, the sampled assignment is the output at root node of $\psi$.

\paragraph*{\textbf{Extending {\inc} to $k$ samples}}
It is straightforward to extend the single sample {\inc} shown in Algorithm~\ref{alg:inc-sampling} to draw $k$ samples in a single pass, where $k$ is a user-specified number. At each node, we have to store a list of $k$ independent copies of partial assignments drawn in $\omega$. At each conjunction node $n_c$, we perform the same union process in line~\ref{line:union-conjunction-child-roulette} of Algorithm~\ref{alg:inc-sampling} for child outputs in the same indices of the respective lists in $\omega$. More specifically, if $n_c$ has child nodes $c_x$ and $c_y$, the outputs of index $i$ are combined to get the output of $n_c$ at index $i$. This process is performed for all indices from 1 to $k$. At each decision node $n_d$, we now draw $k$ independent samples instead of a single sample from the binomial distribution as shown in line~\ref{line:binomial-random}. The sampling step in lines~\ref{line:binomial-random} to~\ref{line:sample-decision-end-roulette} are performed independently for the $k$ random numbers. There is no change necessary for the calculation of joint probabilities in Algorithm~\ref{alg:inc-sampling} as there is no change in literal weights. 

\paragraph*{\textbf{Incremental sampling}}
Given a Boolean formula $F$ and weight function $W$, {\inc} performs incremental sampling with the sampling process shown in Figure~\ref{fig:incremetal-sampling-flow}. In the initial round, {\inc} compiles $F$ and $W$ into a {\prob} $\psi$ and performs sampling. Subsequent rounds involve applying a new set of weights $W$ to $\psi$, typically generated based on existing samples by the controller~\cite{BLM20}, and performing weighted sampling according to the updated weights. The number of sampling rounds is determined by the controller component, whose logic varies according to application. 

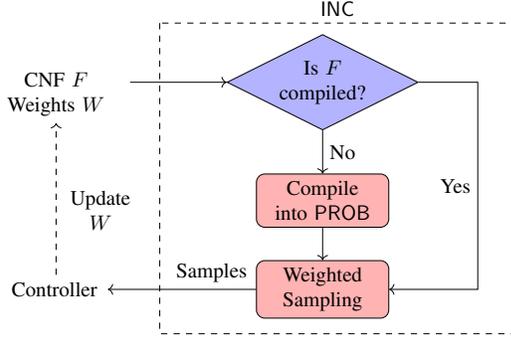
\begin{figure}[htb]
    \centering
    \begin{adjustbox}{width=0.8\columnwidth}
    \begin{tikzpicture}[
    node distance=2cm,
    diamondnode/.style={diamond, minimum width=1.5cm, minimum height=0.5cm, text centered, draw=black, fill=blue!30, aspect=2, inner sep=0pt},
    rectnode/.style={rectangle, rounded corners, minimum width=1.5cm, minimum height=0.5cm,text centered, draw=black, fill=red!30},
    rectnodenofill/.style={rectangle, rounded corners, minimum width=1.5cm, minimum height=0.5cm,text centered, draw=black},
    orangerectnode/.style={rectangle, rounded corners, minimum width=1.5cm, minimum height=0.5cm,text centered, draw=black, fill=orange!30},
    rectnodenofill/.style={rectangle, rounded corners, minimum width=1.5cm, minimum height=0.5cm,text centered, draw=black},
    ]
    \node (compile) [rectnode, text width=2cm] at (1, -2.0) {Compile into {\prob}};
    \node (sample) [rectnode, text width=2cm] at (1, -3.5) {Weighted Sampling};
    \node (iscompiled) [diamondnode, text width=1.75cm] at (1, 0) {Is $F$ compiled?};

    \node (input-cnf)[text width=2.25cm] at (-3.5, 0) {\begin{center} CNF $F$ \\ Weights $W$ \end{center}};
    \node (output) [text width=1.5cm]  at (-3.5, -3.5) {Controller};

    \draw [->] (input-cnf) -- (iscompiled);
    \draw [->] (compile) -- (sample);
    \draw [->] (iscompiled) -- node[anchor=west] {No} (compile);
    \draw [->] (sample) -- node[anchor=south,pos=0.3] {Samples} (output);
    \draw [->] (iscompiled.east) -- +(1, 0) |- +(1, -3.5) node[anchor=east,pos=0.25]{Yes} -- (sample.east);

    \draw [->,dashed] (output) -- node[anchor=west, text width=1.25cm] {\begin{center}Update\\$W$\end{center}} (input-cnf);

    \draw[dashed] (-1.75, 1) rectangle (4.25, -4.25);
    \node (sampler) [] at (1.25, 1.25) {{\inc}};

    \end{tikzpicture}
    \end{adjustbox}

    \caption{{\inc}'s incremental sampling flow}
    \label{fig:incremetal-sampling-flow}
\end{figure}

\subsection{Implementation Decisions} \label{subsec:inc-tech-analysis}

\paragraph*{\textbf{Log-Space Calculations}}
{\inc} performs \textit{annotation} process - computation of joint probabilities in log space. This design choice is made to avoid the usage of arbitrary precision math libraries, which WAPS utilized to prevent numerical underflow after many successive multiplications of probability values. Using the $\mathsf{LogSumExp}$ trick below, it is possible to avoid numerical underflow.

\begin{align*}
    \log(a + b) &= \log(a) + \log( 1 + \frac{b}{a}) \\
    &= \log(a) + \log(1 + \exp( \log(b) - \log(a) ))
\end{align*}

The joint probability at a decision node $n_d$ is given by $\theta_{\lo{n_d}} \times \text{joint probability of } \lo{n_d}  + \theta_{\hi{n_d}} \times \text{joint probability of } \hi{n_d} $. Notice that if we were to perform the calculation in log space, we would have to add the two weighted log joint probabilities, termed $p_{lo}$ and $p_{hi}$ in Algorithm~\ref{alg:inc-sampling}. Using the $\mathsf{LogSumExp}$ trick, we do not need to exponentiate $p_{lo}$ and $p_{hi}$ independently which risks running into numerical underflow. Instead, we only need to exponentiate the difference of $p_{lo}$ and $p_{hi}$ which is more numerically stable. Equations \ref{eq:eq1} and \ref{eq:eq2} can be implemented in log space as follows:
\begin{align*}
    Q \text{ of $\land$-node }n_c &= \sum_{c \in \child{n_c}} Q(c) \\
    Q \text{ of decision-node }n_d &= \mathsf{LogSumExp}[ \\ &\log(\theta_{\lo{n_d}}) + Q(\lo{n_d}), \\
    & \log(\theta_{\hi{n_d}}) + Q(\hi{n_d})]
\end{align*}
In the equations above, $Q$ refers to the corresponding log joint probabilities in \ref{eq:eq1} and \ref{eq:eq2}. In the experiments section, we detail the runtime advantages of using log computations compared to arbitrary precision math computations.

\paragraph*{\textbf{Dynamic Annotation}}
In existing state-of-the-art weighted sampler WAPS, sampling is performed in two passes - the first pass performs \textit{annotation} and the second pass samples assignments according to the joint probabilities. In {\inc}, we combine the two passes into a single bottom-up pass performing \textit{annotation} dynamically while sampling at each node.

\subsection{Theoretical Analysis} \label{subsec:theo-analysis}

\begin{proposition} \label{prop:param-correctness}
    Branch parameters of any decision node $n_d$ are correct sampling probabilities, i.e. $W(x_i) : W(\neg x_i) =  \theta_{\hi{x_i}} : \theta_{\lo{x_i}}$ where $\var{n_d} = x_i$.
\end{proposition}

\begin{proof}
    \begin{align*}
        \frac{W(x_i)}{W(\neg x_i)} &= \frac{\frac{W(x_i)}{W(x_i) + W(\neg x_i)}}{\frac{W(\neg x_i)}{W(x_i) + W(\neg x_i)}} = \frac{\theta_{\hi{x_i}}}{\theta_{\lo{x_i}}}
    \end{align*}
    
    We start with the ratio of literal weights of $x$, multiply both numerator and denominator by $W(x_i) + W(\neg x_i)$ and arrive at the ratio of branch parameters of $n_d$. Notice that only the ratio matters for sampling correctness and not the absolute value of weights.
\end{proof}

\begin{remark}
    Let $n_d$ be an arbitrary decision node in {\prob} $\psi$. When performing sampling according to a weight function $W$, $\theta_{\lo{n_d}}$ is the probability of picking $\neg \var{n_d}$ and $\theta_{\hi{n_d}}$ is that of $\var{n_d}$. The \textit{determinism} property states that the choice of either literal is disjoint at each decision node.
\end{remark}

\begin{proposition} \label{prop:sampling-correctness}
    {\inc} samples an assignment $\tau$ from {\prob} $\psi$ with probability $\frac{1}{N} \prod_{l \in \tau} W(l)$, where $N$ is a normalization factor.
\end{proposition}

\begin{proof}
    The proof consists of two parts, one for $\land$-node and another for decision node.

    \paragraph*{\textbf{$\land$-node}} Let $n_c$ be an arbitrary conjunction node in {\prob} $\psi$. Recall that by decomposability property, $\forall c_i, c_j \in \child{n_c}$ and $c_i \neq c_j$, $\varset{c_i} \cap \varset{c_j} = \emptyset$. As such an arbitrary variable $x_i \in \varset{n_c}$ only belongs to the variable set of one child node $c_i \in \child{n_c}$. Therefore, assignment of $x_i$ can be sampled independent of $x_j$ where $x_j \in \varset{c_j}, \forall c_j \not= c_i$. Let $\tau'_{c_i}$ be partial assignment for child node $c_i \in \child{n_c}$. Notice that each partial assignment $\tau'_{c_i}$ is sampled independently of others as there are no overlapping variables, hence their joint probability is simply the product of their individual probabilities. This agrees with the weight of an assignment being the product of its components, up to a normalization factor.

    \paragraph*{\textbf{Decision node}} Let $n_d$ be an arbitrary decision node in {\prob} $\psi$ and $x_d$ be $\var{n_d}$. At $n_d$, we sample an assignment of $x_d$ based on the parameters $\theta_{\lo{x_d}}$ and $\theta_{\hi{x_d}}$, which are probabilities of literal assignment by Proposition~\ref{prop:param-correctness}. By Proposition~\ref{prop:param-correctness}, one can see that the assignment of $x_d$ is sampled correctly according to $W$. As the sampling process at $n_d$ is independent of its child nodes by the determinism property, the joint probability of sampled assignment of $x_d$ and the output partial assignment from the corresponding child node would be the product of their probabilities. Notice that the joint probability aligns with the definition of weight of an assignment being the product of the weight of its literals, up to a normalization factor.

    Since we do not consider the \textit{false} node and treat it as having 0 probability, we always sample from satisfying assignments by starting at the \textit{true} node in bottom-up ordering. Reconciling the sampling process at the two types of nodes, it is obvious that any combination of decision and $\land$-nodes encountered in the sampling process would agree with a given weight function $W$ up to a normalization factor $1/N$. In fact, $N = \sum_{\tau_i \in S} W(\tau_i)$ where $S$ is the set of satisfying assignments of Boolean formula $F$ that $\psi$ represents. As mentioned in Proposition~\ref{prop:param-correctness} proof, normalization factors do not affect the correctness of sampling according to $W$, and we have shown that {\inc} performs weighted sampling correctly under multiplicative weight functions.
\end{proof}

\begin{remark}
From the proof of Proposition 2, the determinism and decomposability property is important to ensure the correctness of {\inc}. The smoothness property is important to ensure that the sampled assignment by {\inc} is complete. For formula $F = (x \lor y) \land (\neg x \lor \neg z)$, an assignment $\tau_1$ sampled from a non-smooth {\prob} could be $\{x, \neg z \}$. Notice that $\tau_1$ is missing assignment for variable $y$. By performing smoothing, we will be able to sample a complete assignment of all variables in the Boolean formula as both child nodes of each decision node $n$ have the same $\varset{\cdot}$.
\end{remark}

\section{Experiments} \label{sec:experiments}

We implement {\inc} in Python 3.7.10, using NumPy 1.15 and Toposort package. In our experiments, we make use of an off-the-shelf KC diagram compiler, KCBox~\cite{LMY21}. In the later parts of this section, we performed additional comparisons against an implementation of {\inc} using the Gmpy2 arbitrary precision math package ({\incap}) to determine the impact of log-space \textit{annotation} computations.

Our benchmark suite consists of instances arising from a wide range of real-world applications such as DQMR networks, bit-blasted versions of SMT-LIB (SMT) benchmarks, ISCAS89 circuits, and configurable systems~\cite{GSRM19,BLM20}. For incremental updates, we rely on the weight generation mechanism proposed in the context of prior applications of incremental sampling~\cite{BLM20}. In particular, new weights are generated based on the samples from the previous rounds, resulting in the need to recompute joint probabilities in each round. Keeping in line with prior work, we perform 10 rounds (R1-R10) of incremental weighted sampling and 100 samples drawn in each round. The experiments were conducted with a timeout of 3600 seconds on clusters with Intel Xeon Platinum 8272CL processors. 

In this section, we detail the extensive experiments conducted to understand {\inc}'s runtime behavior and to compare it with the existing state-of-the-art weighted sampler WAPS~\cite{GSRM19} in incremental weighted sampling tasks. We chose WAPS as it has been shown to achieve significant runtime improvement over other samplers, and accordingly has emerged as a sampler of the choice for practical applications~\cite{BLM20}.  In particular, our empirical evaluation sought to answer the following questions:

\begin{description}
    \item[RQ 1] How does {\inc}'s incremental weighted sampling runtime performance compare to current state-of-the-art? 
    \item[RQ 2] How does using {\prob} affect runtime performance?
    \item[RQ 3] How does log-space calculations impact runtime performance?
    \item[RQ 4] Does {\inc} correctly perform weighted sampling? 
\end{description}

\begin{figure*}[htb]
    \centering
    \begin{subfigure}[t]{.47\textwidth}
        \centering
        \includegraphics[width=\linewidth]{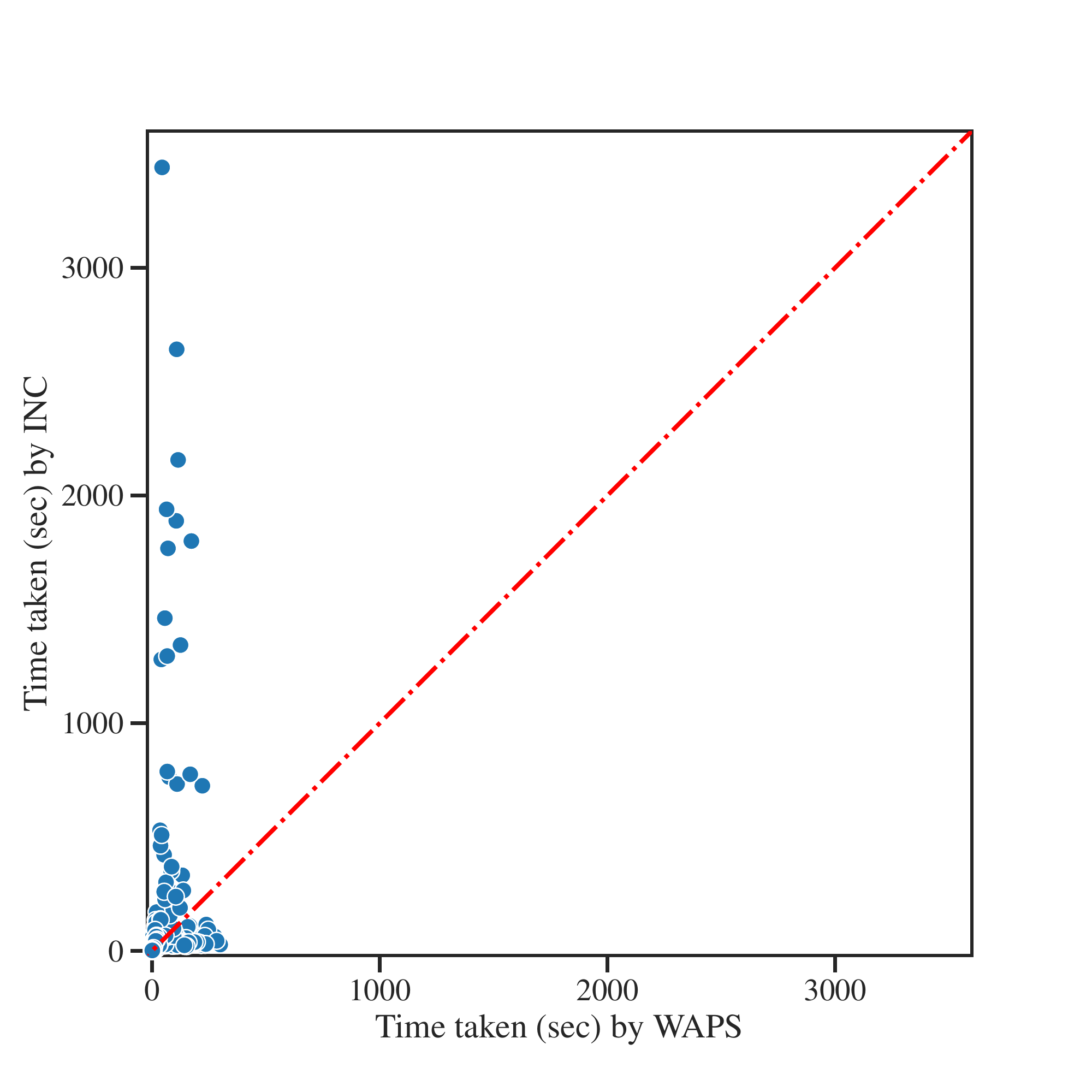}
        \caption{Single Round (R1) Runtime Scatter Plot}
        \label{fig:runtime-comparison-single-round}
    \end{subfigure}%
    \hspace{0.05\linewidth}
    \begin{subfigure}[t]{.47\textwidth}
        \centering
        \includegraphics[width=\linewidth]{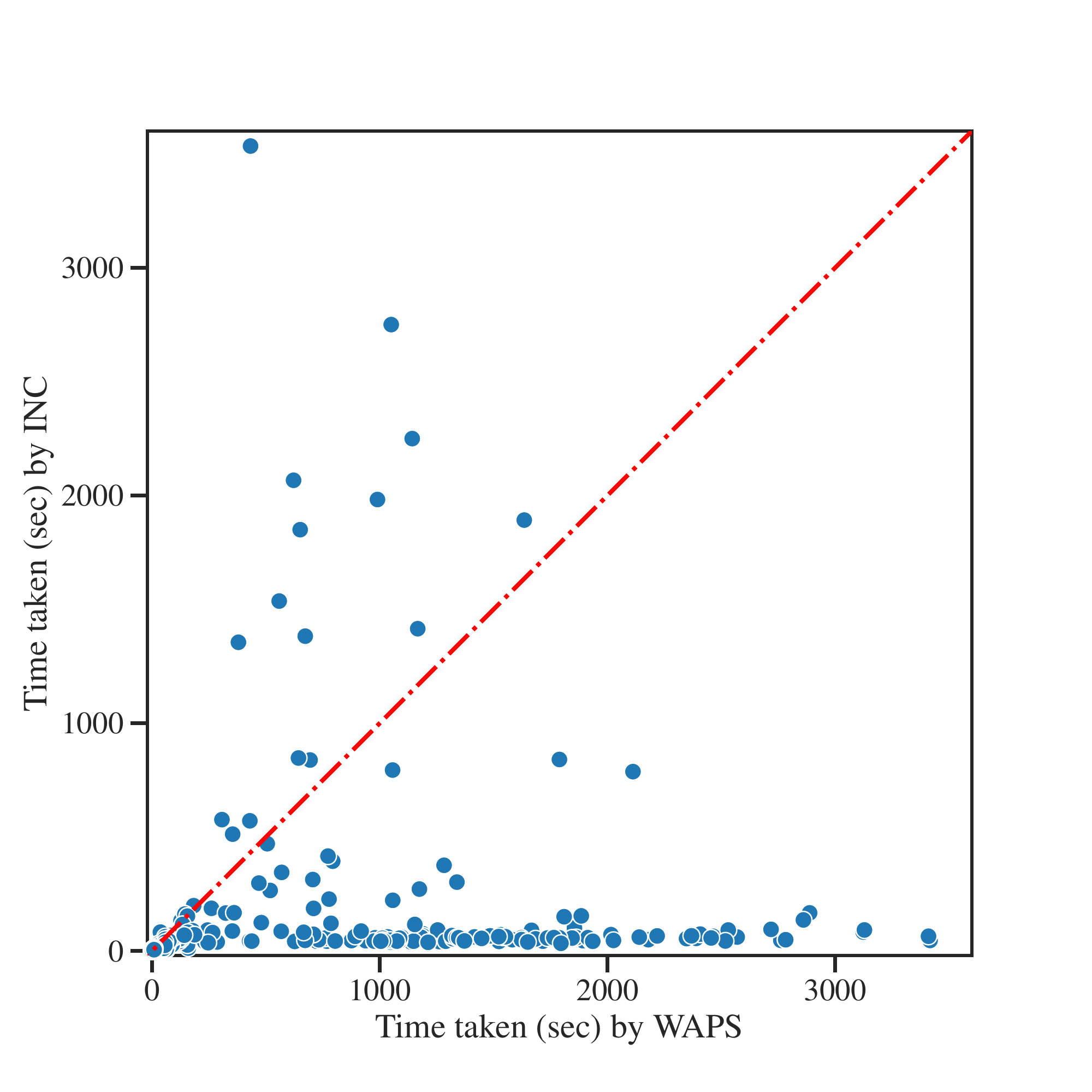}
        \caption{Incremental Runtime Scatter Plot}
        \label{fig:runtime-comparison-incremental}
    \end{subfigure}
    \caption{Runtime comparisons between {\inc} and state-of-the-art weighted sampler WAPS}
    \label{fig:runtime-comparison}
\end{figure*}

\paragraph*{\textbf{RQ 1: Incremental Sampling Performance}}

The scatter plot of incremental sampling runtime comparison is shown in Figure~\ref{fig:runtime-comparison}, with Figure~\ref{fig:runtime-comparison-single-round} showing runtime comparison for the first round (R1) and Figure~\ref{fig:runtime-comparison-incremental} showing runtime comparison over 10 rounds. The vertical axes represent the runtime of {\inc} and the horizontal axes represent that of WAPS. In the experiments, {\inc} completed 650 out of 896 benchmarks whereas WAPS completed 674. {\inc} completed 21 benchmarks that WAPS timed out and similarly, WAPS completed 45 benchmarks that {\inc} timed out. In the experiments, {\inc} achieved a median speedup of $1.69\times$ over WAPS.

\begin{table*}
    \centering
    \begin{tabular}{l|rrrrr}
        \toprule
        Statistic & 
        $\dfrac{\text{WAPS $\mathsf{MEAN}$(R2 to R10)}}{\text{WAPS R1}}$ & 
        $\dfrac{\text{{\inc} $\mathsf{MEAN}$(R2 to R10)}}{\text{{\inc} R1}}$ &
        $\dfrac{\text{WAPS R1}}{\text{{\inc} R1}}$  & 
        $\dfrac{\text{WAPS $\mathsf{SUM}$(R2 to R10)}}{\text{{\inc} $\mathsf{SUM}$(R2 to R10)}}$ & 
        $\dfrac{\text{WAPS Total}}{\text{{\inc} Total}}$
        \\
        \midrule
        Mean & 0.74 & 0.064 & 1.03 & 15.66 & 6.12 \\
        \midrule
        Std & 0.24 & 0.040 & 1.47 & 26.42 & 10.73  \\
        \midrule
        Median & 0.67 & 0.059 & 0.44 & 4.48 & 1.69  \\
        \midrule
        Max & 1.25 & 0.188 & 10.65 & 172.66 & 73.96  \\
        \bottomrule
    \end{tabular}
    \caption{Incremental weighted sampling runtime ratio statistics for WAPS and {\inc} (Numerators and denominators refer to the corresponding runtimes)}
    \label{tab:sampling-result-stats}
\end{table*}

\begin{table*}
    \centering
    \small
    \tabcolsep=0.2cm
    \begin{tabular}{ll|rrrrrrrrrrr|r}
    \toprule
    Benchmark & Tool & R1 & R2 & R3 & R4 & R5 & R6 & R7 & R8 & R9 & R10 & Total & Speed \\
    \midrule
    or-50-5-5-UC-10 & WAPS & \textbf{56.6} & 56.3 & 52.5 & 59.4 & 52.5 & 53.6 & 59.4 & 53.2 & 53.4 & 61.7 & \textbf{558.6} & 1.0$\times$ \\
    (100, 253) & {\inc} & 1461.3 & \textbf{7.6} & \textbf{8.4} & \textbf{8.4} & \textbf{8.4} & \textbf{8.4} & \textbf{8.5} & \textbf{8.5} & \textbf{8.4} & \textbf{8.5} & 1536.3 & 0.4$\times$ \\
    \midrule
    or-100-20-9-UC-30 & WAPS & \textbf{73.0} & 69.1 & 66.7 & 76.0 & 66.5 & 66.9 & 76.6 & 66.0 & 66.9 & 78.6 & 706.1 & 1.0$\times$ \\
    (200, 528) & {\inc} & 269.5 & \textbf{4.7} & \textbf{4.8} & \textbf{4.8} & \textbf{4.9} & \textbf{5.1} & \textbf{4.8} & \textbf{4.8} & \textbf{4.8} & \textbf{5.1} & \textbf{313.4} & 2.3$\times$ \\
    \midrule
    s953a\_15\_7 & WAPS & \textbf{1.7} & 1.1 & 1.1 & 1.2 & 1.0 & 1.1 & 1.2 & 1.1 & 1.1 & 1.3 & 11.9 & 1.0$\times$ \\
    (602, 1657) & {\inc} & 4.9 & \textbf{0.7} & \textbf{0.7} & \textbf{0.7} & \textbf{0.7} & \textbf{0.7} & \textbf{0.7} & \textbf{0.7} & \textbf{0.7} & \textbf{0.7} & \textbf{11.5} & 1.0$\times$ \\
    \midrule
    h8max         & WAPS & 90.3 & 104.2 & 92.4 & 116.0 & 94.3 & 94.1 & 112.9 & 92.9 & 94.4 & 120.4 & 1011.9 & 1.0$\times$           \\
    (1202, 3072)  & {\inc} & \textbf{34.1} & \textbf{2.1} & \textbf{2.2} & \textbf{2.4} & \textbf{2.3} & \textbf{2.4} & \textbf{2.2} & \textbf{2.4} & \textbf{2.4} & \textbf{2.3} & \textbf{55.7} & 18.2$\times$    \\
    \midrule
    innovator      & WAPS & 195.5 & 221.9 & 201.3 & 244.4 & 200.1 & 206.7 & 247.2 & 202.0 & 202.9 & 257.4 & 2179.3 & 1.0$\times$           \\
    (1256, 50452)  & {\inc} & \textbf{32.8} & \textbf{1.6} & \textbf{1.8} & \textbf{1.9} & \textbf{1.9} & \textbf{1.9} & \textbf{1.8} & \textbf{1.9} & \textbf{1.9} & \textbf{1.9} & \textbf{49.4} & 44.1$\times$    \\
    \bottomrule
    \end{tabular}
    \caption{Runtime (seconds) breakdowns for each of ten rounds (R1-R10) between WAPS and {\inc} for benchmarks of different sizes e.g. `h8max' benchmark consists of 1202 variables and 3072 clauses.}
    \label{tab:benchmark-instance-comparison}
\end{table*}

Further results are shown in Table~\ref{tab:sampling-result-stats}. Observe that for runtime taken for R1 (column 3), WAPS is faster and takes around $0.44\times$ of {\inc}'s runtime in the median case. However, {\inc} takes the lead in runtime performance when we examine the total time taken for the incremental rounds R2 to R10 (column 4). For incremental rounds, WAPS always took longer than {\inc}, in the median case WAPS took $4.48\times$ longer than {\inc}. We compare the average incremental round runtime with the first round runtime for both samplers in columns 1 and 2. In the median case, an incremental round for WAPS takes $67\%$ of the time for R1 whereas an incremental round for {\inc} only requires $5.9\%$ of the time R1 takes. We show the per round runtime for 5 benchmarks in Table~\ref{tab:benchmark-instance-comparison} to further illustrate {\inc}'s runtime advantage over WAPS for incremental sampling rounds, even though both tools reuse the respective KC diagram compiled in R1. This set of results highlights {\inc}'s superior performance over WAPS in the handling of incremental sampling settings. {\inc}'s advantage in incremental sampling rounds led to better overall runtime performance than WAPS in $75\%$ of evaluations. The runtime advantage of {\inc} would be more obvious in applications requiring more than 10 rounds of samples.

Therefore, we conducted sampling experiments for 20 rounds to substantiate our claims that {\inc} will have a larger runtime lead over WAPS with more rounds. Both samplers are given the same 3600s timeout as before and are to draw 100 samples per round, for 20 rounds. The number of completed benchmarks is shown in Table~\ref{tab:20round-timeout-comparison} In the 20 sampling round setting, {\inc} completed 649 out of 896 benchmarks, timing out on 1 additional benchmark compared to 10 sampling round setting. In comparison, WAPS completed 596 of 896 benchmarks, timing out on 78 additional benchmarks than in the 10 sampling round setting. In addition, WAPS takes on median $2.17\times$ longer than {\inc} under the 20 sampling round setting, an increase over the $1.69\times$ under the 10 sampling round setting.

\begin{table}
    \centering
    \small
    \begin{tabular}{l|r|r}
        \toprule
        Number of rounds & WAPS & {\inc} \\
        \midrule
        10 & 674 & 650 \\
        \midrule
        20 & 596 & 649 \\
        \bottomrule
    \end{tabular}
    \caption{Number of completed benchmarks within 3600s, for 10 and 20 round settings}
    \label{tab:20round-timeout-comparison}
\end{table}

The runtime results clearly highlight the advantage of {\inc} for incremental weighted sampling applications and that {\inc} is noticeably better at incremental sampling than the current state-of-the-art.

\paragraph*{\textbf{RQ 2: {\prob} Performance Impacts}}

\begin{table}
    \centering
    \small
    \begin{tabular}{l|r}
        \toprule
        Statistic & $\frac{\text{WAPS KC size}}{\text{{\inc} KC size}}$\\
        \midrule
        Mean & 18.92 \\
        \midrule
        Std & 81.19 \\
        \midrule
        Median & 4.64 \\
        \midrule
        Max & 1734.08 \\
        \bottomrule
    \end{tabular}
    \caption{Statistics for number of nodes in d-DNNF (WAPS KC diagram) over that of smoothed {\prob} ({\inc} KC diagram).}
    \label{tab:benchmark-kc-diagram-size}
\end{table}

We now focus on the analysis of the impact of using {\prob} compared to d-DNNF in the design of a weighted sampler. We analyzed the size of both {\prob} and d-DNNF across the benchmarks that both tools managed to compile and show the results in Table~\ref{tab:benchmark-kc-diagram-size}. From Table~\ref{tab:benchmark-kc-diagram-size}, {\prob} is always smaller than the corresponding d-DNNF. Additionally, {\prob} is at median $4.64\times$ smaller than the corresponding d-DNNF, and that for {\prob} is an order of magnitude smaller for at least $25\%$ of the benchmarks. As such, {\prob} emerges as the clear choice of knowledge compilation diagram used in {\inc}, owing to its succinctness which leads to fast incremental sampling runtimes.

\paragraph*{\textbf{RQ 3: Log-space Computation Performance Impacts}}

\begin{table}
    \centering
    \small
    \begin{tabular}{l|r}
        \toprule
        Statistic & $\frac{\text{{\incap} runtime}}{\text{{\inc} runtime}}$\\
        \midrule
        Mean & 1.14 \\
        \midrule
        Std & 0.16 \\
        \midrule
        Median & 1.12 \\
        \midrule
        Max & 1.89 \\
        \bottomrule
    \end{tabular}
    \caption{Runtime comparison of {\inc} and {\incap}}
    \label{tab:inc-implementation-runtime-comparison}
\end{table}

In the design of {\inc}, we utilized log-space computations to perform \textit{annotation} computations as opposed to naively using arbitrary precision math libraries. In order to analyze the impact of this design choice, we implemented a version of {\inc} where the dynamic \textit{annotation} computations are performed using arbitrary precision math in a similar manner as WAPS. We refer to the arbitrary precision math version of {\inc} as {\incap}. As an ablation study, we compare the runtime of both implementations across all the benchmarks and show the comparison in Table~\ref{tab:inc-implementation-runtime-comparison}. The statistics shown is for the ratio of {\incap} runtime to {\inc} runtime, a value of $1.12$ means that {\incap} takes $1.12\times$ that of {\inc} for the corresponding statistics.

The results in Table~\ref{tab:inc-implementation-runtime-comparison} highlight the runtime advantages of our decision to use log-space computations over arbitrary precision computations. {\inc} has faster runtime than {\incap} in majority of the benchmarks. {\inc} displayed a minimum of $0.70 \times$, a median of $1.12 \times$,and a max of $1.89 \times$ speedup over {\incap}. Furthermore, {\incap} timed out on 2 more benchmarks compared to {\inc}. It is worth emphasizing that log-space computations do not introduce any error, and our usage of them sought to improve on the naive usage of arbitrary precision math libraries.   

\paragraph*{\textbf{RQ 4: {\inc} Sampling Quality}}

We conducted additional evaluation to further substantiate evidence of {\inc}'s sampling correctness, apart from theoretical analysis in Section~\ref{subsec:theo-analysis}. Specifically, we compared the samples from {\inc} and WAPS, which has proven theoretical guarantees~\cite{GSRM19}, on the `case110' benchmark that is extensively used by prior works~\cite{SGRM18,AHT18,GSRM19}. We gave each positive literal weight of $0.75$ and each negative literal $0.25$, and subsequently drew one million samples using both {\inc} and WAPS and compare them in Figure~\ref{fig:distplot}.

\begin{figure}[htb]
    \begin{center}
    \centerline{\includegraphics[width=\columnwidth]{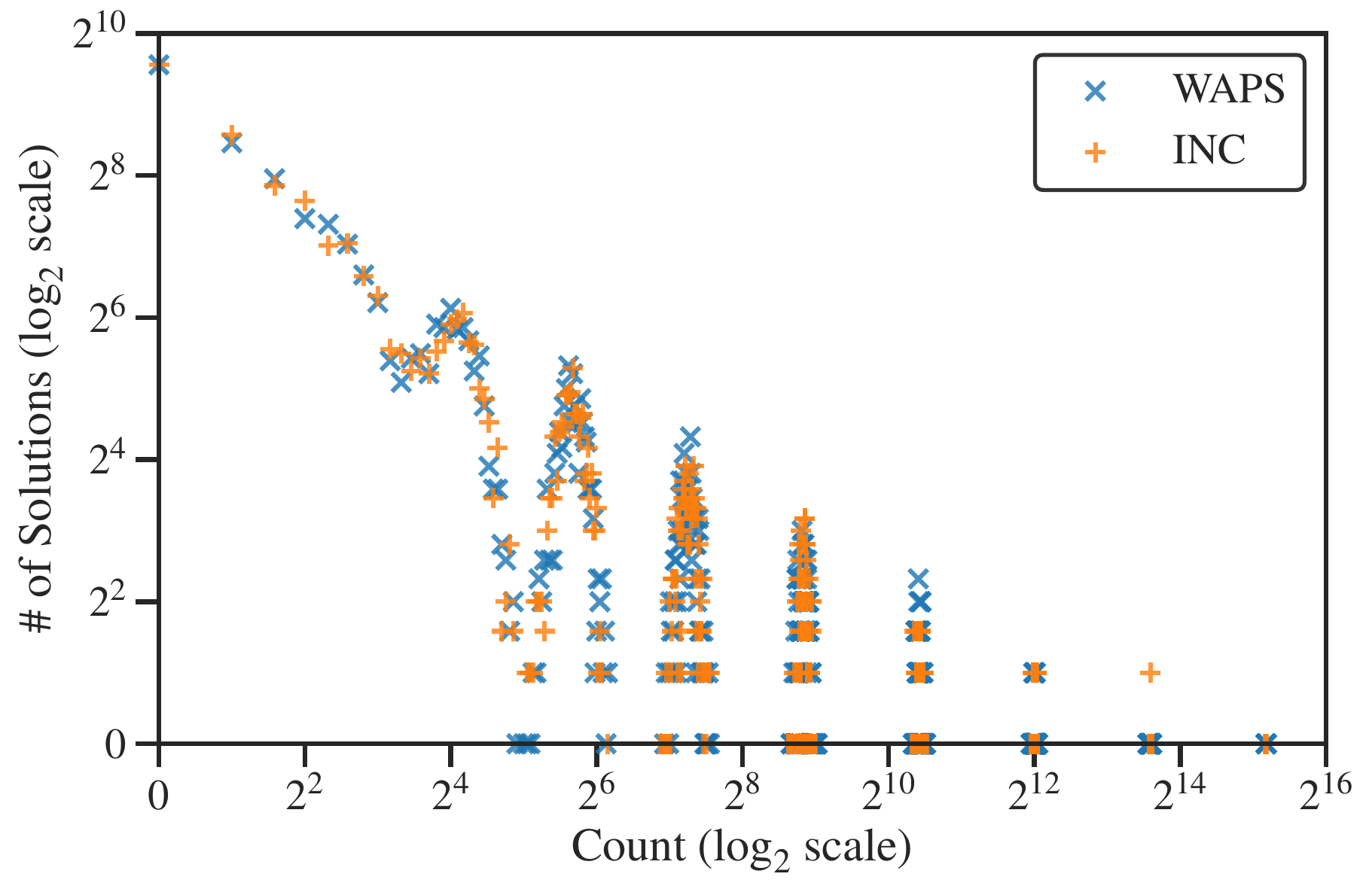}}
    \caption{Distribution comparison for Case110, with $\log$ scale for both axes}
    \label{fig:distplot}
    \end{center}
\end{figure}

Figure~\ref{fig:distplot} shows the distributions of samples drawn by {\inc} and WAPS for `case110' benchmark. A point ($x,y$) on the plot represents $y$ number of unique solutions that were sampled $x$ times in the sampling process by the respective samplers. The almost perfect match between the weighted samples drawn by {\inc} and WAPS, coupled with our theoretical analysis in Section~\ref*{subsec:theo-analysis}, substantiates our claim {\inc}'s correctness in performing weighted sampling. Additionally, it also shows that {\inc} can be a functional replacement for existing state-of-the-art sampler WAPS, given that both have theoretical guarantees.

\paragraph*{\textbf{Discussion}}

We demonstrated the runtime performance advantages of {\inc} and the two main contributing factors - a choice of succinct knowledge compilation form and dynamic log-space \textit{annotation}. {\inc} takes longer than WAPS for single-round sampling, mainly because WAPS takes less time for KC diagram compilation than {\inc}, leading to WAPS being faster in single-round sampling. In the incremental sampling setting, the compilation costs of KC diagrams are amortized, and since {\inc} is substantially better at handling incremental updates, it thus took the overall runtime lead from WAPS in the majority of the benchmarks. Extrapolating the trend, it is most likely that {\inc} would have a larger runtime lead over WAPS for applications requiring more than 10 sampling rounds. The runtime breakdown demonstrates that {\inc} is able to amortize the compilation time over the incremental sampling rounds, with subsequent rounds being much faster than WAPS.
In summary, we show that {\inc} is substantially better at incremental sampling than existing state-of-the-art. 

\section{Conclusion and Future Work} \label{sec:future-works}

In conclusion, we introduced a bottom-up weighted sampler, {\inc}, that is optimized for incremental weighted sampling. By exploiting the succinct structure of {\prob} and log-space computations, {\inc} demonstrated superior runtime performance in a series of extensive benchmarks when compared to the current state-of-the-art weighted sampler WAPS. The improved runtime performance, coupled with correctness guarantees, makes a strong case for the wide adoption of {\inc} in future applications. 

For future work, a natural step would be to seek further runtime improvements for {\prob} compilation since {\inc} takes longer than SOTA for the initial sampling round, due to slower compilation.  Another extension would be to investigate the design of a partial \textit{annotation} algorithm to reduce computations when only a small portion of the weights have been updated. It would also be of interest if we could store partial sampled assignments at each node as a succinct sketch to reduce memory footprint, for instance we could store each unique assignment and its count.

\section*{Acknowledgement}

We sincerely thank Yong Lai for the insightful discussions. Suwei Yang is supported by the Grab-NUS AI Lab, a joint collaboration between GrabTaxi Holdings Pte. Ltd., National University of Singapore, and the Industrial Postgraduate Program (Grant: S18-1198-IPP-II) funded by the Economic Development Board of Singapore. Kuldeep S. Meel is supported in part by National Research Foundation Singapore under its NRF Fellowship Programme (NRF-NRFFAI1-2019-0004), Ministry of Education Singapore Tier 2 grant (MOE-T2EP20121-0011), and Ministry of Education Singapore Tier 1 Grant (R-252-000-B59-114).

\bibliographystyle{ieeetr}
\bibliography{references}

\end{document}